\documentclass[11pt]{amsart}
\usepackage{amsmath}
\usepackage{graphicx}
\theoremstyle{plain}
\newtheorem{theorem}{Theorem}[section]
\newtheorem{lemma}[theorem]{Lemma}
\newtheorem{prop}[theorem]{Proposition}

\theoremstyle{definition}

\numberwithin{equation}{section}

\newcommand{\tr}{\operatorname{Tr}}

\begin{document}

\title[Positive Mass with Angular Momentum and Charge] {The Positive Mass Theorem with Angular Momentum and Charge for Manifolds with Boundary}

\author[Bryden]{Edward T. Bryden}
\author[Khuri]{Marcus A. Khuri}
\author[Sokolowsky]{Benjamin D. Sokolowsky}
\address{Department of Mathematics\\
Stony Brook University\\
Stony Brook, NY 11794, USA}
\email{ebryden@math.sunysb.edu, khuri@math.sunysb.edu, bsokolowsky@math.sunysb.edu}

\thanks{M. Khuri acknowledges the support of NSF Grant DMS-1708798.}

\begin{abstract}
Motivated by the cosmic censorship conjecture in mathematical relativity,
we establish the precise mass lower bound for an asymptotically flat Riemannian 3-manifold with nonnegative scalar curvature and minimal surface boundary, in terms of angular momentum and charge. In particular this result does not require the restrictive assumptions of simple connectivity and completeness, which are undesirable from both a mathematical and physical perspective.
\end{abstract}
\maketitle

\section{Introduction}
\label{sec1} \setcounter{equation}{0}
\setcounter{section}{1}

Heuristic arguments of Penrose \cite{Penrose} which are well-known to motivate the Penrose inequality \cite{Bray,HuiskenIlmanen}, may also be used to obtain a conjectured lower bound for the ADM mass $m$ of a spacetime in terms of total angular momentum $\mathcal{J}$ and charge $Q$, namely
\begin{equation}\label{0}
m^2 \geq \frac{Q^{2}+  \sqrt{Q^4 + 4\mathcal{J}^2}}{2},
\end{equation}
with equality occurring only for the extreme Kerr-Newman black hole. The arguments leading to \eqref{0} require conservation of angular momentum and charge, and for this it is typically assumed that the spacetime is axisymmetric and satisfies certain conditions related to its matter fields.
Inequality \eqref{0} acts as a necessary condition for the grand cosmological censorship conjecture \cite{Penrose1} as well as the final state conjecture \cite{Klainerman}. Therefore while a counterexample would be detrimental for at least one of these conjectures, confirmation under the most general of possible settings only adds to the prevailing belief in their validity. Moreover, inequality \eqref{0} may be viewed as a refinement of the positive mass theorem \cite{SchoenYau0,Witten} in which a precise contribution to the total mass is given in terms of the rotation and charge of black holes.

The appropriate mathematical setting in which to study this inequality is that of an initial data set $(M,g,k)$, consisting of a Riemannian 3-manifold with metric $g$ and a symmetric 2-tensor $k$ representing the second fundamental form of the embedding into spacetime. An asymptotically flat end is required for the definition of ADM mass, and nonnegative scalar curvature is assumed in order that the positive mass theorem is valid. This curvature condition arises from the physical consideration of nonnegative matter energy density together with a maximal slice $\tr k=0$. Through the combined work of several authors \cite{ch1,ChruscielCosta,ChruscielLiWeinstein,Costa,Dain0,KhuriWeinstein,SchoenZhou}, inequality \eqref{0} has been established when the manifold $(M,g)$ is simply connected, complete, and contains another end which is either asymptotically flat or asymptotically cylindrical. The recent survey \cite{DainGabachClement} details many of the developments. The proof follows a two step procedure, the first of which is to obtain an initial lower bound for the mass in terms a renormalized harmonic map energy. The second consists of minimizing this energy, and showing that the unique global minimizer is the singular harmonic map associated with extreme Kerr-Newman data. Simple connectivity is used in the first step to introduce a specialized coordinate system (Brill coordinates) that allows for a simple bulk integral expression of the mass. Completeness and the asymptotics of the other end also play an important role here in that they prevent the appearance of boundary terms in the formula for the mass. Furthermore, simple connectivity is also used in the second step to ensure the existence of a twist potential to efficiently encode angular momentum and construct the harmonic map energy. Thus, these hypotheses play fundamental roles in the proof and whether they can be removed has been unclear. On the other hand, the Penrose arguments motivating the inequality require no such hypotheses, and for this reason it has been conjectured that these assumptions are unnecessary. They are also unnatural as the positive mass theorem itself does not require such restrictions. In particular, it is important to allow the initial data to have minimal surface boundary, as these may be interpreted as cross-sections of the event horizon. Furthermore, significant generalizations of the positive mass theorem, including the positive mass theorem with charge \cite{GHHP} and the Penrose inequality (with charge) \cite{KhuriWeinsteinYamada}, require a minimal surface boundary to be meaningful. Additionally, from a physical perspective it is not desirable to necessitate the presence of a secondary asymptotic end, as this typically represents the interior of a black hole. Indeed, from the point of view of an outside observer, it is not possible to know the structure of spacetime contained within the event horizon. As for simple connectivity, although topological censorship \cite{FriedmanSchleichWitt} implies that this is an appropriate assumption for initial data within the domain of outer communication, it says nothing about the fundamental group of the interior black hole region. In fact, it suggests that all nontrivial topology is contained within the black hole, and therefore the combined assumptions of simple connectivity, completeness, and the existence of a secondary asymptotically flat end are not physically justified.

The purpose of the present work is to establish \eqref{0} in generality without the unwanted hypotheses discussed above, for a single black hole. We also obtain a mass lower bound in the multi-black hole case consistent with that proved under the more restrictive hypotheses in \cite{ChruscielLiWeinstein,KhuriWeinstein}.

Let vector fields $E$ and $B$ defined on $M$ represent the electric and magnetic fields. It will be assumed that both are divergence free, that is, there is no charged matter. Traces of the Gauss and Codazzi equations for the embedding of the initial data into spacetime yield formulas for the energy and momentum densities of the matter fields minus electromagnetic contributions via the constraint equations
\begin{equation}\label{1}
16\pi\mu_{em} = R+(\tr k)^{2}-|k|^{2}-2(|E|^{2}+|B|^{2}),
\end{equation}
\begin{equation}
8\pi J_{em} = \operatorname{div}(k-(\tr k)g)+2E\times B,
\end{equation}
where $R$ is scalar curvature and $(E\times B)_{i}=\epsilon_{ijl}E^{j}B^{l}$ denotes the cross product in which $\epsilon$ is the volume form of $g$. We will say that the initial data are axially symmetric if there exists a $U(1)$ subgroup within the group of isometries of the Riemannian manifold $(M,g)$, and all relevant quantities are invariant
under the $U(1)$ action. In particular, the following Lie derivatives vanish
\begin{equation}\label{2}
\mathfrak{L}_{\eta}g=\mathfrak{L}_{\eta}k=\mathfrak{L}_{\eta}E=\mathfrak{L}_{\eta}B=0,
\end{equation}
where $\eta$ is the Killing field corresponding to the symmetry.

The initial data will be referred to as asymptotically flat if there exists an end $M_{\text{end}}\subset M$ diffeomorphic to $\mathbb{R}^{3}\setminus\text{Ball}$, such that with respect to the asymptotic coordinates
\begin{equation}\label{falloff1}
g_{ij}=\delta_{ij}+o_{\ell}(r^{-\frac{1}{2}}),\text{ }\text{ }\text{ }\text{ }\partial g_{ij}\in
L^{2}(M_{\text{end}}),\text{
}\text{ }\text{ }
\text{ }k_{ij}=O_{\ell-1}(r^{-\lambda}),\text{ }\text{ }\text{ }\text{ }\mu_{em}, J_{em}^{i},J_{em}(\eta)\in L^{1}(M_{\text{end}}),
\end{equation}
\begin{equation}\label{falloff2}
E^{i}=O_{\ell-1}(r^{-\lambda}),\text{ }\text{ }\text{ }\text{ }\text{ }B^{i}=O_{\ell-1}(r^{-\lambda}),\text{
}\text{ }\text{ }
\text{ }\lambda>\frac{3}{2},
\end{equation}
for some $\ell\geq 5$.
These fall-off conditions ensure that the ADM mass and angular momentum, as well as the total electric and magnetic charge are well-defined by
\begin{equation}\label{admmass}
m=\frac{1}{16\pi}\int_{S_{\infty}}(g_{ij,i}-g_{ii,j})\nu^{j},
\end{equation}
\begin{equation}\label{admam}
\mathcal{J}=\frac{1}{8\pi}\int_{S_{\infty}}(k_{ij}-(\tr k)g_{ij})\nu^{i}\eta^{j},
\end{equation}
\begin{equation}\label{charge}
Q_{e} = \frac1{4\pi} \int_{S_\infty} E_i \nu^i\, , \qquad
Q_{b} = \frac1{4\pi} \int_{S_\infty} B_i \nu^i\, .
\end{equation}
Here limits as $r\rightarrow\infty$ for integrals taken over coordinate spheres $S_r$, with unit outer normal $\nu$, are represented by $S_{\infty}$. The total charge is the combination of the total electric and magnetic charges $Q^2=Q_{e}^2+Q_{b}^2$. It should also be pointed out that \eqref{admmass} is referred to as mass by an abuse of terminology, as this expression is technically the ADM energy or time component in the ADM 4-momentum. Furthermore an electromagnetic contribution to the definition of angular momentum is sometimes included \cite{DainKhuriWeinsteinYamada}, although the two definition agree under the current asymptotics.

In order for the Penrose heuristic arguments to be valid an energy condition is needed so that Hawking's area theorem \cite{Hawking} may be utilized, and hypotheses must be imposed which guarantee the conservation of angular momentum and charge. This motivates the following assumptions used in the proof of \eqref{0} from the initial data perspective. For the energy condition, nonnegative nonelectromagnetic matter energy density $\mu_{em}\geq 0$ will be assumed. This is consistent with the dominant energy condition typically used in association with the positive mass theorem (including charge), as well as previous inequalities involving angular momentum. Furthermore, axisymmetry and the electrovacuum assumption $\mu_{em}=|J_{em}|=0$ have been used in the past \cite{Dain} with regards to conservation of angular momentum and charge. However, here it is only necessary to require vanishing of one component of the momentum density, namely in the Killing direction $J_{em}(\eta)=0$, and the divergence free property for the Maxwell field. Inequality \eqref{0} is known to be false without the assumption of axisymmetry \cite{HuangSchoenWang}.

\begin{theorem}\label{massam}
Let $(M,g,k,E,B)$ be an axisymmetric, maximal initial data set for the Einstein-Maxwell equations with one asymptotically flat end, minimal surface boundary, and satisfying $\mu_{em}\geq 0$ in addition to $J_{em}(\eta)=\operatorname{div}E=\operatorname{div}B=0$. If either
\begin{itemize}
\item [(i)] the outermost minimal surface has a single component, or\smallskip

\item [(ii)] the boundary $\partial M$ has one component and $M$ is simply connected,
\end{itemize}
then
\begin{equation}\label{3}
m^2>\frac{Q^{2}+\sqrt{Q^{4}+4\mathcal{J}^{2}}}{2}.
\end{equation}
\end{theorem}

The first point to note is that there is a strict inequality in \eqref{3}. This is to be expected since from the heuristic physical arguments leading to \eqref{0}, equality should only be achieved if the initial data agree with the canonical slice of an extreme Kerr-Newman spacetime. However the extreme Kerr-Newman data do not possess a minimal surface, but rather have a cylindrical end, and therefore do not satisfy the hypotheses of Theorem \ref{massam}. The minimal surface boundary, which could consist of many components, together with the asymptotically flat end guarantee the existence of an outermost minimal surface \cite{Eichmair}, and the assumption that it has one component is analogous to the case of the Penrose inequality treated by Huisken and Ilmanen \cite{HuiskenIlmanen}. In order to treat \eqref{3} in the presence of a multicomponent outermost minimal surface it is most likely that new ideas will be needed, as was the case for the multiple black hole version of the Penrose inequality established by Bray \cite{Bray}. It is interesting to note that unlike the Penrose inequality, \eqref{3} continues to hold if the boundary $\partial M$ is merely a single component minimal surface but not necessarily outer area minimizing. In fact under this assumption treated in $(ii)$, one can drop the hypothesis on the outermost minimal surface and replace it with simple connectivity of $M$ to obtain the same conclusion.

The proof of Theorem \ref{massam} is based on a doubling procedure in so called pseudospherical coordinates, where the data are reflected across the outermost minimal surface. This requires axisymmetry of the outermost minimal surface, a fact that does not appear to exist within the literature and is thus proven below in Proposition \ref{axi}. This result is of independent interest as it may be applied elsewhere, for example to extensions of the Penrose inequality that include contributions from angular momentum.

We are able to extend Theorem \ref{massam} to allow for certain types of multiple black holes by including a mixture of boundary components and extra asymptotically flat ends as well as asymptotically cylindrical ends (see \cite{KhuriWeinstein}). However, as in the case of a complete, simply connected initial data set, the presence of multiple black holes does not immediately yield an explicit expression for the mass lower bound \cite{ChruscielLiWeinstein,KhuriWeinstein}. Rather, the lower bound is given in terms of the reduced harmonic energy of a Weinstein stationary solution \cite{weinstein96} to the Einstein-Maxwell equations having the same angular momentum and charge for each black hole. This harmonic energy is denoted by $\mathcal{F}$, and is a function of the angular momenta and charge. It is conjectured that the resulting inequality coincides with the expression \eqref{3} in which $\mathcal{J}$ and $Q$ are the sums of the angular momenta and charge from the different horizon components.

\begin{theorem}\label{thm2}
Let $(M,g,k,E,B)$ be an axisymmetric, maximal, asymptotically flat initial data set for the Einstein-Maxwell equations having a minimal surface boundary and a finite number of additional ends each of which is asymptotically flat or asymptotically cylindrical. Assume further that $\mu_{em}\geq 0$ in addition to $J_{em}(\eta)=\operatorname{div}E=\operatorname{div}B=0$.
If either
\begin{itemize}
\item [(i)] at most one component of the outermost minimal surface encloses components of the boundary $\partial M$ or nonsimply connected domains, or\smallskip

\item [(ii)] the boundary $\partial M$ has one component and $M$ is simply connected,
\end{itemize}
then
\begin{equation}\label{3.1}
m\geq\mathcal{F}(\mathcal{J}_{1},\ldots,\mathcal{J}_{N},Q_{e}^{1},\ldots,Q_{e}^{N},
Q_{b}^{1},\ldots,Q_{b}^{N})
\end{equation}
where $N$ is the combined number of additional ends and components of $\partial M$ and $\mathcal{J}_i$, $Q_{e}^i$, $Q_{b}^i$ represent the angular momentum and charge associated with each of these ends and boundary components.
\end{theorem}


\section{The Doubling Procedure in Pseudospherical Coordinates}
\label{sec2} \setcounter{equation}{0}
\setcounter{section}{2}

Consider the setting of case $(ii)$ in Theorem \ref{massam} where $(M,g)$ is axisymmetric, asymptotically flat, and simply connected with a single component minimal surface boundary. It follows from \cite[Theorem 2.2]{ChruscielNguyen} that $M$ is diffeomorphic to $\mathbb{R}^3\setminus B_{m_{1}/2}(0)$, and there exists a global system of cylindrical-type coordinates $(\rho,z,\phi)$ on this domain such that the metric takes the form
\begin{equation}\label{metric}
g=e^{-2U+2\alpha}(d\rho^2+dz^2)+\rho^2 e^{-2U}(d\phi+A_{\rho} d\rho +A_{z} dz)^2,
\end{equation}
where the Killing field $\eta=\partial_{\phi}$. The isothermal part of \eqref{metric} is the metric on the orbit space $M/U(1)$, and the remaining part arises from the dual 1-form $\rho^2 e^{-2U}(d\phi +A_{\rho} d\rho +A_{z} dz)$ to the Killing field. This structure for $g$ in pseudospherical coordinates is related to the Weyl-Papapetrou form \cite{ChruscielNguyen} used in the reduction of the stationary axisymmetric vacuum Einstein equations. With the standard transformation $\rho=r\sin\theta$, $z=r\cos\theta$ producing spherical-type coordinates $(r,\theta,\phi)$ with ranges $m_{1}/2\leq r<\infty$, $0\leq\theta\leq \pi$, and $0\leq\phi<2\pi$, the fall-off of the metric coefficients in the asymptotically flat end is given by
\begin{equation}\label{f}
U=o_{\ell-3}(r^{-\frac{1}{2}}),\quad \alpha=o_{\ell-4}(r^{-\frac{1}{2}}),\quad
A_{\rho}=\rho o_{\ell-3}(r^{-\frac{5}{2}}),\quad A_{z}=o_{\ell-3}(r^{-\frac{3}{2}}).
\end{equation}
Furthermore $\alpha=0$ on the axis $\rho=0$, and all coefficients are independent of $\phi$. Note also that the value $m_1>0$ is uniquely determined, and the existence of pseudospherical coordinates does not require the boundary $\partial M$ to be minimal. Since the mean curvature of a coordinate sphere $S_r$ is
\begin{equation}
H=\frac{2/r+\partial_{r}(\alpha-2U)}{\sqrt{e^{-2U+2\alpha}+\rho^2 e^{-2U} A_{r}^2}}
\end{equation}
where $A_{r}=\sin\theta A_{\rho}+\cos\theta A_{z}$, the assumption of a minimal boundary $\partial M=S_{m_{1}/2}$ is equivalent to
\begin{equation}\label{minimal}
\partial_{r}\left(U-\frac{1}{2}\alpha\right)=\frac{2}{m_1}.
\end{equation}

A particularly advantageous feature of the metric structure \eqref{metric} is the simple expression obtained for the scalar curvature \cite{Dain0}
\begin{equation}\label{scalar}
2e^{-2U+2\alpha}R=8\Delta U-4\Delta_{\rho,z}\alpha-4|\nabla U|^{2}-\rho^{2}e^{-2\alpha}
\left(A_{\rho,z}-A_{z,\rho}\right)^{2},
\end{equation}
where $\Delta$ is the Laplacian on $\mathbb{R}^{3}$ with respect to the flat metric $\delta=d\rho^2+dz^2+\rho^2 d\phi^2$ and $\Delta_{\rho,z}=\partial_{\rho}^{2}+\partial_{z}^{2}$. Moreover the constraint equation \eqref{1}, and the assumptions of a maximal slice $\tr k=0$ and nonnegative energy density $\mu_{em}\geq 0$ imply that
\begin{align}\label{scalar1}
\begin{split}
R=&16\pi\mu_{em}+|k|^{2}+2\left(|E|^{2}+|B|^{2}\right)\\
\geq
&2\frac{e^{6U-2\alpha}}{\rho^{4}}|\nabla v
+\chi\nabla\psi-\psi\nabla\chi|^{2}
+2\frac{e^{4U-2\alpha}}{\rho^{2}}\left(|\nabla\chi|^{2}
+|\nabla\psi|^{2}\right),
\end{split}
\end{align}
where $v$, $\chi$, and $\psi$ are potential functions for angular momentum, electric charge, and magnetic charge respectively. More precisely, the divergence free property of the electric and magnetic fields combined with Cartan's magic formula shows that the 1-forms $\iota_{\eta}\star E$ and $\iota_{\eta}\star B$ are closed, where $\iota$ and $\star$ denote interior product and the Hodge star operation. Hence simple connectivity yields global potentials satisfying
\begin{equation}\label{empotential}
d\chi=\iota_{\eta}\star E,\quad\quad\quad d\psi=\iota_{\eta}\star B.
\end{equation}
Furthermore, as shown in \cite{KhuriWeinstein} the 1-form $\star( k(\eta)\wedge \eta)-\chi d\psi+\psi d\chi$ is closed exactly when $J_{em}(\eta)=0$. Therefore under the hypotheses of Theorem \ref{massam} there exists a global twist potential satisfying
\begin{equation}\label{twistpotential}
dv=\star( k(\eta)\wedge \eta)-\chi d\psi+\psi d\chi.
\end{equation}
The inequality in \eqref{scalar1} then follows in a straightforward way from \eqref{empotential} and \eqref{twistpotential}.  Moreover, if
$\omega=dv+\chi d\psi-\psi d\chi$ then asymptotics \cite{KhuriWeinstein} for the potentials are expressed by
\begin{equation}\label{asym}
|\omega|=\rho^{2}O(r^{-\lambda}),\text{ }\text{ }\text{ }\text{ }|\nabla\chi|+|\nabla\psi|=\rho O(r^{-\lambda})\text{
}\text{ }\text{ as }\text{ }\text{ }r\rightarrow\infty,
\end{equation}
\begin{equation}\label{asym.1}
|\omega|=O(\rho^{2}),\text{ }\text{ }\text{ }\text{ }|\nabla\chi|+|\nabla\psi|=O(\rho) \text{ }\text{ }\text{ as }\text{
}\text{ }\rho\rightarrow 0\text{ }\text{ }\text{ in }\text{ }\text{ }\mathbb{R}^3\setminus B_{m_1/2}(0).
\end{equation}
In addition, since $|\eta|=0$ on the $z$-axis all the potential functions are constant there, and the difference of these constants associated with the two connected components $I_{+}=\{\rho=0, z> m_{1}/2\}$ and $I_{-}=\{\rho=0, z<-m_{1}/2\}$ of the axis yield the angular momentum and charges
\begin{equation}\label{potentialcharge}
\mathcal{J}=\frac{1}{4}\left(v|_{I_{-}}-v|_{I_{+}}\right),\quad\quad
Q_e=\frac{1}{2}\left(\chi|_{I_{-}}-\chi|_{I_{+}}\right),\quad\quad
Q_b=\frac{1}{2}\left(\psi|_{I_{-}}-\psi|_{I_{+}}\right).
\end{equation}

Typically a mass lower bound in terms of a harmonic map energy is obtained by
integrating \eqref{scalar} over $M$ and applying a version of \eqref{scalar1}. This works well when $(M,g)$ is complete, however here the presence of a boundary leads to boundary terms which are not desirable when minimizing the harmonic map energy. Therefore we seek to double the manifold across its boundary, and show that the same strategy may be carried out on the doubled manifold with two ends. The primary difficulty arises from the lack of regularity across the doubling surface. Nevertheless we show that the minimal surface hypothesis is sufficient for the argument to go through.

Consider the conformal map $f:B_{m_1/2}\setminus\{0\}\rightarrow \mathbb{R}^3\setminus B_{m_1/2}$ given by spherical inversion
\begin{equation}
f(\tilde{r},\tilde{\theta},\tilde{\phi})=\left(\left(\frac{m_1}{2}\right)^2\frac{1}{\tilde{r}},
\tilde{\theta},\tilde{\phi}\right),
\end{equation}
which is expressed in cylindrical coordinates as
\begin{equation}\label{b}
\rho=\left(\frac{m_1}{2}\right)^2\frac{\tilde{\rho}}{\tilde{r}^2},\quad\quad\quad
z=\left(\frac{m_1}{2}\right)^2\frac{\tilde{z}}{\tilde{r}^2},\quad\quad\quad
\phi=\tilde{\phi}.
\end{equation}
Pulling back the metric to $B_{m_1/2}\setminus\{0\}$ yields
\begin{equation}\label{newmetric}
\tilde{g}:=f^{*}g=e^{-2\tilde{U}+2\tilde{\alpha}}
(d\tilde{\rho}^2+d\tilde{z}^2)+\tilde{\rho}^2 e^{-2\tilde{U}}(d\tilde{\phi}
+\tilde{A}_{\rho} d\tilde{\rho} +\tilde{A}_{z} d\tilde{z})^2,
\end{equation}
where
\begin{equation}\label{a}
\tilde{U}=2\log\tilde{r}+2\log(2/m_1)+U\circ f,\quad\quad\quad\quad
\tilde{\alpha}=\alpha\circ f,
\end{equation}
\begin{equation}
\tilde{A}_{\rho}=\tilde{r}^{-4}\left(\frac{2}{m_1}\right)^2
\left[(\tilde{z}^2-\tilde{\rho}^2)A_{\rho}-2\tilde{\rho}A_{z}\right],\quad\quad
\tilde{A}_{z}=\tilde{r}^{-4}\left(\frac{2}{m_1}\right)^2
\left[(\tilde{\rho}^2-\tilde{z}^2)A_{z}-2\tilde{z}A_{\rho}\right].
\end{equation}
This leads to a metric and potentials globally defined on the complement of the origin
\begin{equation}
\bar{g}=
  \begin{cases}
  g & \text{on $\mathbb{R}^3\setminus B_{m_1/2}$,} \\
  \tilde{g} & \text{on $B_{m_1/2}\setminus\{0\}$.}
  \end{cases}
\end{equation}
Similarly, the potentials may also be extended to the ball by setting $\tilde{v}=v\circ f$, $\tilde{\chi}=\chi\circ f$, $\tilde{\psi}=\psi\circ f$ in $B_{m_1/2}\setminus\{0\}$, and the corresponding functions defined on $\mathbb{R}^3\setminus\{0\}$ will be denoted $\bar{v}$, $\bar{\chi}$, and $\bar{\psi}$. These functions and the metric $\bar{g}$ are $C^{0,1}$ and smooth away from the reflection sphere $S_{m_1/2}$.

The form of the metric \eqref{newmetric} guarantees that the scalar curvature of $\bar{g}$ satisfies the equation \eqref{scalar} on all of $\mathbb{R}^3\setminus\{0\}$. Moreover it also satisfies the lower bound in \eqref{scalar1}. To see this observe that
\begin{equation}
|\nabla\chi|^2\circ f=
(\partial_{\rho}\chi)^2\circ f+(\partial_{z}\chi)^2\circ f
=\left(\frac{2}{m_1}\right)^4\tilde{r}^4\left[(\partial_{\tilde{\rho}}\tilde{\chi})^2
+(\partial_{\tilde{z}}\tilde{\chi})^2\right]
=\left(\frac{2}{m_1}\right)^4\tilde{r}^4|\tilde{\nabla}\tilde{\chi}|^2,
\end{equation}
and similarly
\begin{equation}
|\nabla v+\chi\nabla\psi-\psi\nabla\chi|^2\circ f
=\left(\frac{2}{m_1}\right)^4\tilde{r}^4|\tilde{\nabla} \tilde{v}+\tilde{\chi}\tilde{\nabla}\tilde{\psi}-\tilde{\psi}\tilde{\nabla}
\tilde{\chi}|^2.
\end{equation}
Combining this with \eqref{scalar1}, \eqref{b}, and \eqref{a} shows that
in $B_{m_1/2}\setminus\{0\}$
\begin{align}\label{scalarfinal}
\begin{split}
\tilde{R}=&R\circ f\\
\geq
&2\frac{e^{6U\circ f-2\alpha\circ f}}{(\rho\circ f)^{4}}|\nabla v
+\chi\nabla\psi-\psi\nabla\chi|^{2}\circ f
+2\frac{e^{4U\circ f-2\alpha\circ f}}{(\rho\circ f)^{2}}\left(|\nabla\chi|^{2}\circ f
+|\nabla\psi|^{2}\circ f\right)\\
=& 2\frac{e^{6\tilde{U}-2\tilde{\alpha}}}{\tilde{\rho}^{4}}|\tilde{\nabla} \tilde{v}
+\tilde{\chi}\tilde{\nabla}\tilde{\psi}-\tilde{\psi}\tilde{\nabla}\tilde{\chi}|^{2}
+2\frac{e^{4\tilde{U}-2\tilde{\alpha}}}{\tilde{\rho}^{2}}
\left(|\tilde{\nabla}\tilde{\chi}|^{2}
+|\tilde{\nabla}\tilde{\psi}|^{2}\right).
\end{split}
\end{align}
It follows that the scalar curvature of the doubled metric $\bar{g}$ satisfies the desired lower bound on $\mathbb{R}^3\setminus\{0\}$ away from the sphere $S_{m_1/2}$.
Although the metric is not sufficiently regular across this sphere to have a pointwise defined scalar curvature on this surface, the fact that it is a minimal surface with respect to both inner and outer domains guarantees that $\bar{R}$ satisfies the inequality distributionally. Furthermore the minimal surface property allows for the fundamental mass lower bound in terms of scalar curvature, despite the lack of metric regularity.

In order to establish the mass lower bound it is necessary to note that the doubled manifold $(\bar{M},\bar{g})$, where $\bar{M}=\mathbb{R}^{3}\setminus\{0\}$, possesses two asymptotically flat ends. Indeed, at the additional end near the origin the metric coefficients and potentials satisfy the asymptotics
\begin{equation}\label{asym1}
\bar{U}=2\log r+C+o_{1}(r^{\frac{1}{2}}),\text{ }\text{ }\text{ }\text{ }\bar{\alpha}=o_{1}(r^{\frac{1}{2}}),\text{
}\text{ }\text{ }\text{ }\bar{A}_{\rho}=\rho o_{1}(r^{-\frac{5}{2}}),\text{ }\text{ }\text{ }\text{
}\bar{A}_{z}=o_{1}(r^{-\frac{3}{2}}),
\end{equation}
\begin{equation}\label{asym2}
|\bar{\omega}|=\rho^{2}O(r^{\lambda-6}),\text{ }\text{ }\text{ }|\nabla\bar{\chi}|+|\nabla\bar{\psi}|=\rho O(r^{\lambda-4})\text{
}\text{ }\text{ as }\text{ }\text{ }r\rightarrow 0,
\end{equation}
for some constant $C$. Here and in what follows, unless stated otherwise, the tilde notation will be removed from coordinates within the domain $B_{m_1/2}\setminus\{0\}$.

\begin{lemma}\label{lemma1}
The doubled manifold $(\bar{M},\bar{g})$ possesses two asymptotically flat ends, and the mass is given by
\begin{equation}\label{masslower}
m
=\frac{1}{32\pi}\int_{\mathbb{R}^{3}}\left(2e^{-2\bar{U}+2\bar{\alpha}}\!\text{ }\bar{R}
+4|\nabla \bar{U}|^{2}
+\rho^{2}e^{-2\bar{\alpha}}
(\bar{A}_{\rho,z}-\bar{A}_{z,\rho})^{2}
\right)dx,
\end{equation}
where $dx$ is the Euclidean volume element.
\end{lemma}

\begin{proof}
Although the metric $\bar{g}$ is only Lipschitz across $S_{m_1/2}$, the fact that this sphere is a minimal surface guarantees that a particular combination of coefficients has improved regularity, namely $\bar{U}-\frac{1}{2}\bar{\alpha}\in C^{1,1}$. Moreover, this is all that is needed to establish \eqref{masslower}.

First observe that in light of \eqref{minimal}
\begin{equation}\label{agree1}
\lim_{r\rightarrow \frac{m_1}{2}^+} \partial_{r}\left(\bar{U}-\frac{1}{2}\bar{\alpha}\right)=\frac{2}{m_1}.
\end{equation}
It suffices then to show that the limit from inside $B_{m_1/2}$ yields the same value. For emphasis we will use the tilde notation to perform this computation. By
\eqref{a}
\begin{equation}
\partial_{\tilde{r}}\tilde{U}=\frac{2}{\tilde{r}}+\partial_{r}U\frac{\partial r}{\partial\tilde{r}}
=\frac{2}{\tilde{r}}-\left(\frac{m_1}{2}\right)^2 \frac{1}{\tilde{r}^2}\partial_{r}U.
\end{equation}
Therefore \eqref{minimal} implies
\begin{equation}
\partial_{\tilde{r}}\tilde{U}|_{\tilde{r}=\frac{m_1}{2}}=
\frac{4}{m_1}-\partial_{r}U|_{r=\frac{m_1}{2}}
=\frac{2}{m_1}-\frac{1}{2}\partial_{r}\alpha|_{r=\frac{m_1}{2}}.
\end{equation}
On the other hand
\begin{equation}
\partial_{\tilde{r}}\tilde{\alpha}=-\left(\frac{m_1}{2}\right)^2 \frac{1}{\tilde{r}^2}\partial_{r}\alpha,
\end{equation}
and therefore the desired conclusion follows
\begin{equation}\label{agree}
\lim_{r\rightarrow \frac{m_1}{2}^-} \partial_{r}\left(\bar{U}-\frac{1}{2}\bar{\alpha}\right)=\frac{2}{m_1}.
\end{equation}

We will now show that \eqref{masslower} holds. According to \cite{ch1}
\begin{equation}\label{fk}
m=\lim_{r\to \infty}\frac{1}{4\pi}\left(\int_{S_r}\partial_r\left(\bar{U}
-\frac{1}{2}\bar{\alpha}\right)+\frac{1}{2}\int_{W_r}\frac{\bar{\alpha}}{\rho} \right),
\end{equation}
where $W_r=\{\rho=r, -r<z<r\}$ is the wall of the cylinder or radius $r$.
Next observe that \eqref{agree1} and \eqref{agree} yield
\begin{align}\label{ibp}
\begin{split}
\int_{S_r}\partial_r\left(\bar{U}-\frac{1}{2}\bar{\alpha}\right) =&\int_{B_r\setminus B_{m_1/2}}\Delta\left(\bar{U}-\frac{1}{2}\bar{\alpha}\right)dx +\lim_{r\rightarrow \frac{m_1}{2}^+}\int_{S_{r}}\partial_{r}\left( \bar{U}-\frac{1}{2}\bar{\alpha}\right)
\\
=&\int_{B_r\setminus B_{m_1/2}}\Delta\left(\bar{U}-\frac{1}{2}\bar{\alpha}\right)dx +
\lim_{r\rightarrow \frac{m_1}{2}^-}\int_{S_{r}}\partial_{r}\left( \bar{U}-\frac{1}{2}\bar{\alpha}\right)
\\
=&\int_{B_r\setminus B_{m_1/2}}\Delta\left(\bar{U}-\frac{1}{2}\bar{\alpha}\right)dx
+\int_{B_{m_1/2}}\Delta\left(\bar{U}-\frac{1}{2}\bar{\alpha}\right)dx,
\end{split}
\end{align}
since the asymptotics \eqref{asym1} show that
\begin{equation}
\lim_{r\rightarrow 0}\int_{S_{r}}\partial_{r}\left( \bar{U}-\frac{1}{2}\bar{\alpha}\right)=0.
\end{equation}
Moreover, since $\bar{\alpha}$ is continuous across $S_{m_1/2}$, vanishes away from the origin on the $z$-axis, and satisfies \eqref{asym1}
\begin{equation}\label{al}
\lim_{r\rightarrow\infty}\int_{W_r}\frac{\bar{\alpha}}{\rho}
=\int_{\mathbb{R}^3}\frac{1}{\rho}\partial_{\rho}\bar{\alpha}dx.
\end{equation}
Finally, since \eqref{scalar} holds globally on $\bar{M}$ we have
\begin{equation}\label{uu}
\Delta\left(\bar{U}-\frac{1}{2}\bar{\alpha}\right)
+\frac{1}{2\rho}\partial_{\rho}\bar{\alpha}
=\Delta\bar{U}-\frac{1}{2}\Delta_{\rho,z}\bar{\alpha}
=\frac{1}{4}e^{-2\bar{U}+2\bar{\alpha}}\bar{R}
+\frac{1}{2}|\nabla\bar{U}|^2
+\frac{1}{8}\rho^2 e^{-2\bar{\alpha}}
\left(\bar{A}_{\rho,z}-\bar{A}_{z,\rho}\right)^2.
\end{equation}
The desired mass formula \eqref{masslower} now follows by combining
\eqref{fk}, \eqref{ibp}, \eqref{al}, and \eqref{uu}.
\end{proof}

Lemma \ref{lemma1} relates the mass to an energy functional with the help of
\eqref{scalar1} and \eqref{scalarfinal}. Namely together they imply
\begin{equation}\label{190}
m\geq\mathcal{I}(\Psi),
\end{equation}
where $\Psi=(\bar{U},\bar{v},\bar{\chi},\bar{\psi})$ and
\begin{equation}\label{energy}
\mathcal{I}(\Psi)
=\frac{1}{8\pi}\int_{\mathbb{R}^{3}}\left(|\nabla \bar{U}|^{2}
+\frac{e^{4\bar{U}}}{\rho^{4}}|\nabla \bar{v}
+\bar{\chi}\nabla\bar{\psi}-\bar{\psi}\nabla\bar{\chi}|^{2}
+\frac{e^{2\bar{U}}}{\rho^{2}}\left(|\nabla\bar{\chi}|^{2}
+|\nabla\bar{\psi}|^{2}\right)\right)dx.
\end{equation}
The functional $\mathcal{I}$ may be interpreted as the reduced harmonic energy \cite{KhuriWeinstein} for maps $\Psi:\mathbb{R}^3\setminus\{0\}\rightarrow
\mathbb{H}_{\mathbb{C}}^2$ into the complex hyperbolic plane.
Note that the asymptotics \eqref{f}, \eqref{asym}, \eqref{asym1}, and \eqref{asym2} guarantee that $\mathcal{I}(\Psi)$ is finite precisely when $\lambda>\frac{3}{2}$.

\begin{proof}[Proof of Theorem \ref{massam} $(ii)$]
Since the map $\Psi$ is smooth away from the sphere $S_{m_1/2}$, Lipschitz across this surface, and satisfies the asymptotics \eqref{f}, \eqref{asym}, \eqref{asym.1}, \eqref{asym1}, \eqref{asym2}, the gap bound of Schoen and Zhou \cite{SchoenZhou} applies to yield
\begin{equation}\label{gapbound}
\mathcal{I}(\Psi)-\mathcal{I}(\Psi_{0})
\geq C\left(\int_{\mathbb{R}^{3}}
\operatorname{dist}_{\mathbb{H}_{\mathbb{C}}^{2}}^{6}(\Psi,\Psi_{0})dx
\right)^{\frac{1}{3}},
\end{equation}
where $\Psi_0$ is the renormalized harmonic map associated with the extreme Kerr-Newman black hole possessing the same angular momentum and charge as $\Psi$,
and $\operatorname{dist}_{\mathbb{H}_{\mathbb{C}}}$ denotes distance in the complex hyperbolic plane. In particular, together with \eqref{190} we obtain
\begin{equation}
m\geq\mathcal{I}(\Psi_0).
\end{equation}
The desired inequality \eqref{3} now follows since
\begin{equation}\label{kerrnewman}
\mathcal{I}(\Psi_0)^2=\frac{Q^{2}+\sqrt{Q^{4}+4\mathcal{J}^{2}}}{2}.
\end{equation}

Consider now the case in which equality holds in \eqref{3}. This implies, with the help of \eqref{190}, \eqref{gapbound}, and \eqref{kerrnewman}, that $\Psi=\Psi_0$ and in particular $U=U_0$. However this is a contradiction since the asymptotics  \eqref{asym1} show that $U=2\log r+O(1)$ as $r\rightarrow 0$, whereas the corresponding asymptotics for the extreme Kerr-Newman map are given by $U_0=\log r +O(1)$. This difference arises from the fact that $\Psi$ arises from an asymptotically flat geometry near the origin, while the extreme Kerr-Newman initial data possess instead an asymptotically cylindrical end in this location.
\end{proof}

\section{The Outermost Minimal Surface in Axisymmetry}
\label{sec3} \setcounter{equation}{0}
\setcounter{section}{3}

Let $(M,g)$ be as in case $(i)$ in Theorem \ref{massam}. Since this manifold is asymptotically flat and possesses a minimal surface boundary, there exists a unique outermost minimal surface \cite{Eichmair,HuiskenIlmanen} which is a compact embedded smooth hypersurface $\Sigma$. The term outermost refers to the fact that there are no other minimal surfaces homologous to $\Sigma$ which lie outside it with respect to the asymptotic end. The set $M\setminus \Sigma$ consists of one unbounded component, and perhaps several bounded components the union of which will be denoted by $\Omega$, so that $\partial\Omega=\Sigma$.  Each component of the outermost minimal surface must be a topological 2-sphere \cite{Galloway}, since $\Sigma$ is outer area minimizing in that it has the least area among all surfaces which enclose it. Furthermore, according to \cite[Lemma 4.1]{HuiskenIlmanen} $M\setminus \Sigma$ is diffeomorphic to the complement of a finite number of open 3-balls in $\mathbb{R}^3$ with disjoint closure. Here we show that the property of axisymmetry for the ambient manifold descends to $\Sigma$.

\begin{prop}\label{axi}
If $(M,g)$ is axisymmetric then the outermost minimal surface is also axisymmetric.
\end{prop}

\begin{proof}
Suppose that the outermost minimal surface $\Sigma$ is not axisymmetric. Let $\varphi_t$ denote the flow of the axisymmetric Killing field $\eta$, so that $\partial_t\varphi_t=\eta\circ\varphi_t$. Lack of axisymmetry implies that $\eta$ is not tangential to $\Sigma$ at all points. Therefore there exists a nonzero $t_0\sim 0$ such that a domain within $\varphi_{t_0}(\Sigma)$ lies outside of $\Sigma$. Note that $\varphi_t$ is a flow by isometries so $\varphi_{t_0}(\Sigma)$ is a minimal surface, and it is still an embedded 2-sphere.

Let $\mathcal{S}$ denote the compact set which is the union of all compact immersed minimal surfaces within $M$, and define the trapped region $\mathcal{T}$ to be the union of $\mathcal{S}$ with all the bounded components of $M\setminus\mathcal{S}$.
The trapped region is compact and its topological boundary is comprised of embedded smooth minimal 2-spheres \cite[Lemma 4.1]{HuiskenIlmanen}. In fact the outermost minimal surface arises as the boundary $\partial\mathcal{T}$.
In light of this, and the fact that a portion of $\varphi_{t_0}(\Sigma)$ lies outside $\Sigma$, it follows that $\Sigma\neq\partial\mathcal{T}$ contradicting the uniqueness of the outermost minimal surface.
\end{proof}

\begin{proof}[Proof of Theorem \ref{massam} $(i)$]
Assume that the outermost minimal surface $\Sigma$ has a single component. From the discussion above we then have that $M_0 =M\setminus\Omega$ is diffeomorphic to the complement of a 3-ball in $\mathbb{R}^3$. Proposition \ref{axi} guarantees that $\partial M_0$ is axisymmetric, and hence $M_0$ is axisymmetric. It follows that $(M_0,g,k,E,B)$ satisfies the hypotheses of Theorem \ref{massam} $(ii)$, and has the same mass, angular momentum, and charge as the original data.
Part $(ii)$ may now be applied to obtain \eqref{3}.
\end{proof}

\section{Multiple Black Holes}
\label{sec4} \setcounter{equation}{0}
\setcounter{section}{4}

Consider the setting of case $(ii)$ in Theorem \ref{thm2} where $(M,g)$ is axisymmetric, asymptotically flat, simply connected, with a single component minimal surface boundary, and a finite number $n$ of additional ends each of which is asymptotically flat or asymptotically cylindrical; see \cite{KhuriWeinstein} for a definition of asymptotically cylindrical ends. The additional ends may be interpreted physically as individual black holes. A version of pseudospherical coordinates exists for this situation, where each additional end is represented by a puncture on the $z$-axis and again the boundary component is identified with a coordinate sphere.

\begin{prop}
Under the hypotheses of Theorem \ref{thm2} $(ii)$, $M$ is diffeomorphic to $\left(\mathbb{R}^3\setminus B_{m_1/2}(0)\right)\setminus\cup_{i=1}^{n}\{p_i\}$ and there exists a global system of cylindrical-type coordinates $(\rho,z,\phi)$ such that $g$ takes the form \eqref{metric} with $\alpha=0$ whenever $\rho=0$. Each puncture $p_i$ represents an additional end in which the metric coefficients have the asymptotics \eqref{asym1} in the asymptotically flat case, or
\begin{equation}\label{asym3}
U=\log r_i +O_{1}(1),\text{ }\text{ }\text{ }\text{ }\alpha=o_{1}(r_{i}^{\frac{1}{2}}),\text{
}\text{ }\text{ }\text{ }A_{\rho}=\rho o_{1}(r_i^{-\frac{5}{2}}),\text{ }\text{ }\text{ }\text{
}A_{z}=o_{1}(r_i^{-\frac{3}{2}}),
\end{equation}
in the asymptotically cylindrical case. Here $r_i$ denotes the Euclidean distance to the puncture $p_i$.
\end{prop}

\begin{proof}
The proof is nearly identical to that of \cite[Theorem 2.2]{ChruscielNguyen}, and thus we only give an outline. Since $M$ is simply connected the single boundary component $\partial M$ must topological be a 2-sphere by \cite[Lemma 4.9]{Hempel}.
The boundary may then be filled in with a 3-ball, and the metric extended to this domain to obtain a complete, axisymmetric, simply connected Riemannian manifold $(\hat{M},\hat{g})$ with $n+1$ asymptotic ends. According to \cite{ch1,Sokolowsky} $M$ is diffeomorphic to $\mathbb{R}^3\setminus\cup_{i=1}^{n}\{p_i\}$ with the punctures $p_i$ lying on the $\hat{z}$-axis of a global system of cylindrical (Brill) coordinates $(\hat{\rho},\hat{z},\phi)$ in which $\hat{g}$ has the structure \eqref{metric}. The orbit space $M/U(1)$ is identified with the $\hat{\rho}\hat{z}$-half plane, and may be doubled across the axis so that $(\hat{\rho},\hat{z})$ parameterize $\mathbb{R}^2$ minus the axis punctures. Within this plane the projection of $\partial M$ is given by a smooth closed curve $\gamma$ which intersects the $\hat{z}$-axis at two points, and bounds a disc. Using the Riemann mapping theorem, a conformal transformation of the plane may now be applied which maps $\gamma$ to a circle centered at the origin of radius $m_1/2$. The new coordinates obtained from this map are the desired pseudospherical coordinates $(\rho,z,\phi)$. Although the punctures may move under this transformation, they will remain on the axis since the mapping is axisymmetric. Lastly, the conformal property of the map preserves the metric structure \eqref{metric}.
\end{proof}

Simple connectivity of $M$ yields potentials $v$, $\chi$, and $\psi$ satisfying \eqref{empotential}, \eqref{twistpotential} as well as the asymptotics \eqref{asym}, \eqref{asym2} in the asymptotically flat ends. In asymptotically cylindrical ends \cite{KhuriWeinstein}
\begin{equation}\label{asymcyl}
|\omega|=\rho^{2}O(r_i^{\lambda-5}),\text{ }\text{ }\text{ }\text{ }|\nabla\chi|+|\nabla\psi|=\rho O(r_i^{\lambda-3})\text{
}\text{ }\text{ as }\text{ }\text{ }r_i\rightarrow 0.
\end{equation}
The punctures $\{p_i\}_{i=1}^n$ and ball $B_{m_1/2}$ break up the $z$-axis into a sequence of connected component intervals $\{I_j\}_{j=1}^{n+2}$ on which each of the potentials is constant; this is sometimes referred to as a `rod structure'. As in \eqref{potentialcharge}, the difference of two such constants associated to the intervals surrounding a puncture yields the angular momentum or charge associated to the black hole represented by the puncture.

\begin{figure}
\includegraphics[width=13cm]{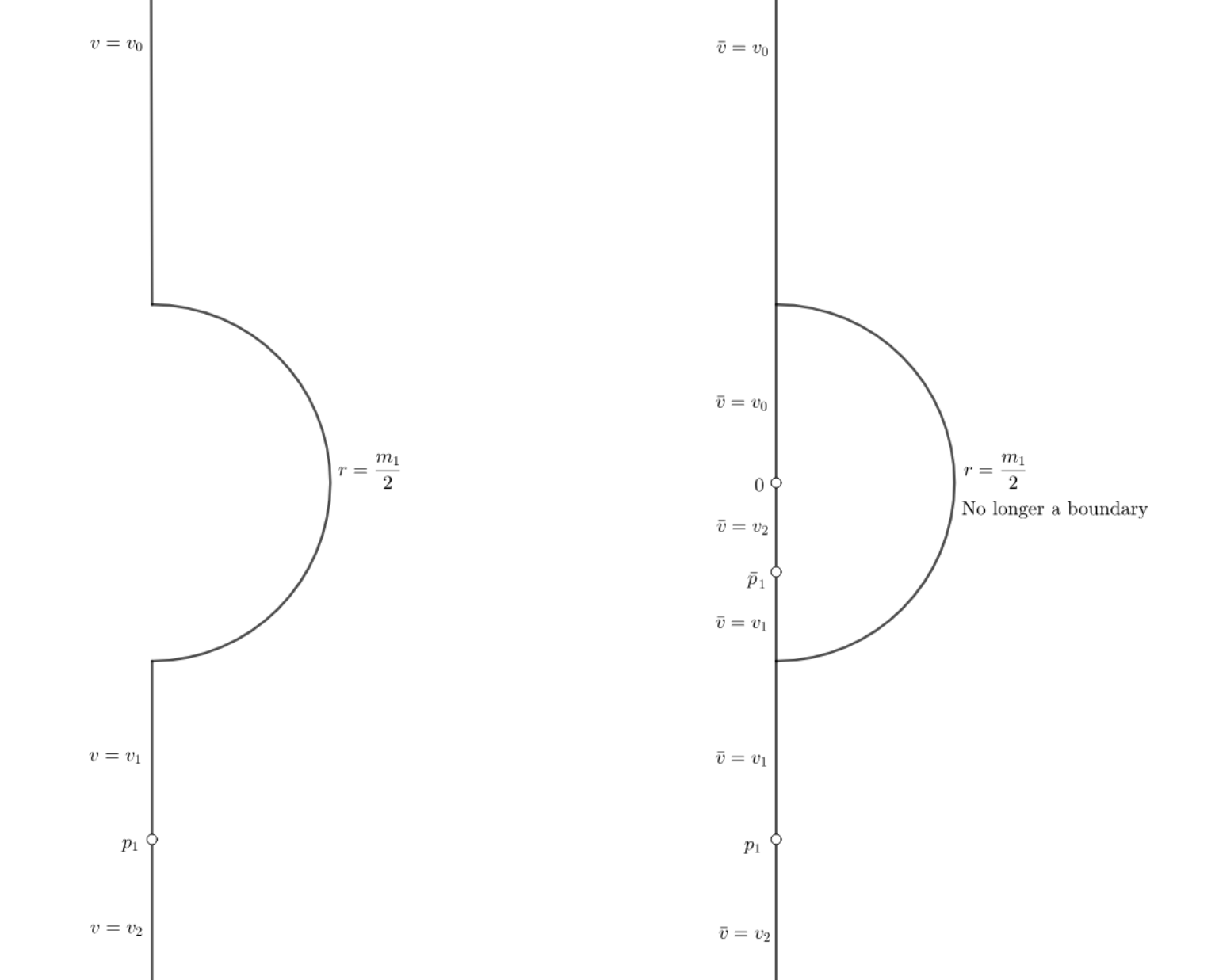}
\caption{Doubling with additional ends}  \label{fig}
\end{figure}

Following Section \ref{sec2} we may double the manifold and potentials across the sphere $S_{m_1/2}$ to obtain a manifold $(\bar{M},\bar{g})$ and functions $\bar{v}$, $\bar{\chi}$, $\bar{\psi}$. The only difference that occurs concerns the number of asymptotic ends. Previously the new manifold had two asymptotically flat ends, however now $\bar{M}$ has $2n+2$ asymptotic ends. In Figure \ref{fig} a diagram of the doubling rod structure in the orbit space is shown, where a single additional end occurs below the circle of radius $m_1/2$ at point $p_1$. After doubling, this end is reflected inside the circle to point $\bar{p}_1$ which represents another end of the same asymptotic type. As before the origin also becomes an asymptotically flat end, associated with the designated asymptotically flat end at infinity. Notice, as is shown in the diagram, that the potential constants on the axis also reflect inside in such a way that they are smooth across the sphere $S_{m_1/2}$ and so that the angular momentum and charge of each end inside the $B_{m_1/2}$ has the same value with opposite sign as that associated with the corresponding puncture outside the ball. Therefore the total angular momentum and total charge of $\bar{M}$, computed by adding all the individual contributions of each end, agrees with the total angular momentum and total charge of $M$.

The presence of additional asymptotically flat and cylindrical ends does not affect the proof of Lemma \ref{lemma1}, as well as scalar curvature lower bounds \eqref{scalar1} and \eqref{scalarfinal}. It follows that as before
\begin{equation}
m\geq\mathcal{I}(\Psi).
\end{equation}

\begin{proof}[Proof of Theorem \ref{thm2}]
Consider case $(ii)$. It remains to show that the renormalized energy may be minimized by a harmonic map
\begin{equation}\label{879}
\mathcal{I}(\Psi)\geq\mathcal{I}(\Psi_1).
\end{equation}
Here $\Psi_1$ is the unique renormalized harmonic map from $\mathbb{R}^3\setminus\{\text{$z$-axis}\}\rightarrow \mathbb{H}^{2}_{\mathbb{C}}$, having the same potential constants as $\Psi$. Such solutions have been constructed in \cite{KhuriWeinstein}, and the corresponding gap bound \eqref{gapbound} was established there as well. Thus, by setting $\mathcal{I}(\Psi_1)=\mathcal{F}$ we obtain the desired result.

Consider now case $(i)$. As in the proof of Theorem \ref{massam} $(i)$ let $M_0$ denote the region exterior to the outermost minimal surface $\Sigma$, with respect to the designated asymptotically flat end. Then $M_0$ is diffeomorphic to the complement of a finite number of open 3-balls and a finite number of points in $\mathbb{R}^3$, where the point punctures represent asymptotically cylindrical ends and the boundary of the 3-balls are the components of $\Sigma$. By assumption, at most one component $\Sigma_1$ of the outermost minimal surface encloses components of the boundary $\partial M$ or nonsimply connected domains. If $\Sigma_1=\emptyset$ then $M$ is simply connected and has no boundary, and therefore this theorem follows from \cite{KhuriWeinstein}. If $\Sigma_1\neq\emptyset$ let $M_1$ denote the region of $M$ outside of $\Sigma_1$. The hypotheses then imply that $M_1$ has a single component minimal surface boundary, is simply connected, and has a finite number of additional asymptotically flat and cylindrical ends. By Proposition \ref{axi} $\partial M_1=\Sigma_1$ is axisymmetric, so that $M_1$ itself is axisymmetric. The initial data $(M_1,g,k,E,B)$ now satisfy the hypotheses of Theorem \ref{thm2} $(ii)$, and \eqref{3.1} follows.
\end{proof}


\begin{thebibliography}{99}

\bibitem{Bray} H. Bray, \textit{Proof of the Riemannian Penrose inequality using the positive mass theorem}, J. Differential Geom., \textbf{59} (2001), 177-267.









\bibitem{ch1} P. Chru\'{s}ciel, \textit{Mass and angular-momentum inequalities for axi-symmetric initial data sets. I. Positivity of Mass}, Ann. Phys., \textbf{323} (2008), 2566-2590.


\bibitem{ChruscielCosta} P. Chru\'{s}ciel, and J. Costa, \textit{Mass, angular-momentum and charge inequalities for axisymmetric initial data}, Classical Quantum Gravity, \textbf{26} (2009), no. 23, 235013.


\bibitem{ChruscielLiWeinstein} P. Chru\'{s}ciel, Y. Li, and G. Weinstein, \textit{Mass and angular-momentum inequalities for axi-symmetric initial data sets. II. Angular Momentum}, Ann. Phys., \textbf{323} (2008), 2591-2613.

\bibitem{ChruscielNguyen} P. Chru\'{s}ciel, and L. Nguyen, \textit{A lower bound for the mass of axisymmetric connected black hole data sets}, Classical Quantum Gravity, \textbf{28} (2011), 125001.



\bibitem{Costa} J. Costa, \textit{Proof of a Dain inequality with charge}, J. Phys. A, \textbf{43} (2010), no. 28, 285202.

\bibitem{Dain0} S. Dain, \textit{Proof of the angular momentum-mass inequality for axisymmetric black hole}, J. Differential Geom., \textbf{79} (2008), 33-67.

\bibitem{Dain} S. Dain, \textit{Geometric inequalities for axially symmetric black holes}, Classical Quantum Gravity, \textbf{29} (2012), no. 7, 073001.

\bibitem{DainGabachClement} S. Dain, and M. Gabach-Clement, \textit{Geometrical inequalities bounding angular momentum and charges in general relativity},
    Living Rev. Relativ., (2018) \textbf{21}, no. 5.




\bibitem{DainKhuriWeinsteinYamada} S. Dain, M. Khuri, G. Weinstein, and S. Yamada,
\textit{Lower bounds for the area of black holes in terms of mass, charge, and angular momentum}, Phys. Rev. D, \textbf{88} (2013), 024048.

\bibitem{Eichmair}  M. Eichmair, \textit{Existence, regularity, and properties of generalized apparent horizons}, Comm. Math. Phys., \textbf{294} (2010), no. 3, 745-760.

\bibitem{FriedmanSchleichWitt} J. Friedman, K. Schleich, and D. Witt, \textit{Topological censorship}, Phys. Rev. Lett., \textbf{71} (1993), no. 10, 1486-1489.


\bibitem{Galloway} G. Galloway, \textit{Stability and rigidity of extremal surfaces in Riemannian geometry and general relativity}, Surveys in geometric analysis and relativity, 221-239, Adv. Lect. Math. (ALM), \textbf{20}, Int. Press, Somerville, MA, 2011.

\bibitem{GHHP} G. Gibbons, S. Hawking, G. Horowitz, and M. Perry,
\textit{Positive mass theorem for black holes}, Comm. Math. Phys., \textbf{88} (1983), 295-308.



\bibitem{Hawking} S. Hawking, and G. Ellis, \textit{The Large Structure of Space-Time}, Cambridge Monographs on Mathematical Physics, Cambridge University Press, 1973.

\bibitem{Hempel} J. Hempel, \textit{3-manifolds}, Annals of Mathematics Studies No. 86, Princeton University Press, 1976.


\bibitem{HuangSchoenWang} L.-H. Huang, R. Schoen, M.-T. Wang, \textit{Specifying angular momentum and center of mass for vacuum initial data sets}, Comm. Math. Phys., \textbf{306} (2011), no. 3, 785-803.

\bibitem{HuiskenIlmanen} G. Huisken, and T. Ilmanen, \textit{The inverse mean curvature flow and the Riemannian Penrose inequality}, J. Differential Geom., \textbf{59} (2001), 353-437.





\bibitem{Sokolowsky} M. Khuri, and B. Sokolowsky, \textit{Lower bounds for the ADM mass of manifolds with cylindrical ends}, preprint, 2018.


\bibitem{KhuriWeinstein}  M. Khuri, and G. Weinstein, \textit{The Positive mass theorem for multiple rotating charged black holes}, Calc. Var. Partial Differential Equations, \textbf{55} (2016), no. 2, 1-29.

\bibitem{KhuriWeinsteinYamada} M. Khuri, G. Weinstein, and S. Yamada, \textit{Proof of the Riemannian Penrose inequality with charge for multiple black holes}, J. Differential Geom., \textbf{106} (2017), 451-498.

\bibitem{Klainerman} S. Klainerman, \textit{Mathematical challenges of general relativity}, Rendiconti di Matematica, Serie VII, \textbf{27} (2007), 105-122.




\bibitem{Penrose} R. Penrose, \textit{Naked singularities}, Ann. New York Acad. Sci., \textbf{224} (1973), 125-134.

\bibitem{Penrose1} R. Penrose, \textit{Some unsolved problems in classical general relativity}, Seminar on Differential Geometry, Ann. Math. Study, \textbf{102} (1982), 631-668.

\bibitem{SchoenYau0} R. Schoen, and S.-T. Yau, \textit{On the proof of the
positive mass conjecture in general relativity}, Comm. Math. Phys., \textbf{65} (1979), no. 1, 45-76.

\bibitem{SchoenZhou} R. Schoen, and X. Zhou, \textit{Convexity of reduced energy and mass angular momentum inequalities}, Ann. Henri Poincar\'{e}, \textbf{14} (2013), 1747-1773. arXiv:1209.0019.





\bibitem{weinstein96} G. Weinstein, \textit{Harmonic maps with prescribed singularities into Hadamard manifolds}, Math. Res. Lett.,
    \textbf{3} (1996), no 6, 835-844.


\bibitem{Witten} E. Witten, \textit{A new proof of the positive energy theorem}, Comm. Math. Phys., \textbf{80} (1981), 381-402.


\end{thebibliography}
\end{document}